\documentclass{article}
\usepackage[T2A]{fontenc}
\usepackage[cp1251]{inputenc}
\usepackage[english,russian]{babel}
\usepackage[tbtags]{amsmath}
\usepackage{amsfonts,amssymb,mathrsfs,amscd}
\usepackage[hyper]{msb-a}
\JournalName{}
\numberwithin{equation}{section}
\theoremstyle{plain}
\newtheorem{theorem}{Theorem}
\newtheorem{lemma}{Lemma}
\newtheorem{alg}{Algorithm}

\theoremstyle{definition}
\newtheorem{proof}{Proof}
\newtheorem{remark}{Remark}

\begin{document}

\title{A new algorithm for solving the $rSUM$ problem}
\author{Valerii Sopin}
\address{Lomonosov Moscow State University, Moscow}
\email{VvS@myself.com}
\udk{510.52} \maketitle
\begin{fulltext}

\begin{abstract}
A determined algorithm is presented for solving the $rSUM$ problem
for any natural $r$ with a sub-quadratic assessment of time
complexity in some cases. In terms of an amount of memory used the
obtained algorithm is the $n\log^3 n$ order.
\end{abstract}

\markright{A new algorithm for solving the $rSUM$ problem}

\section{Introduction}\label{s1}
In computational complexity theory, the $3SUM$ problem asks if a
given set of $n$ integers, each with absolute value bounded by some
polynomial in $n$, contains three elements that sum to
zero.~\cite{1, 2}. The generalized version, $rSUM$, asks the same
question for $r$ elements.~\cite{1, 2}.

The $3SUM$ problem was initially set in~\cite{1}. Gajentaan and
Overmars collected a large list of geometric problems, which may be
solved in an order of quadratic complexity, and nobody knows, how to
do it faster~\cite{1}.

Hereinafter, we understand the order of complexity as asymptotic
complexity of the algorithm, namely: the computational complexity
(number of operations) of a given algorithm is bounded from above
with function $f(n)$ (which is the order of complexity) with
accuracy to the constant multiplier and for the sufficiently large
input length $n$.

The $3SUM$ problem has a simple and obvious algorithm for solving in
the order of $n^2$ operations~\cite{1, 2}.

There are a probabilistic, sub-quadratic algorithms~\cite{3} in the
computational model, which implies parallel memory operation.

A determined algorithm of solving the $3SUM$ problem based on the
Fast Fourier Transformation was suggested in~\cite{4}. However it
assumes that absolute values of these $n$ numbers are limited by the
number $\frac{n^2}{\log n}$.

There are a algorithms based on sorting with partial
information~\cite{5}.

A solution to the generalized version of the problem, $rSUM$, may be
found in~\cite{2}. Its known order of complexity is
$n^{\frac{r}{2}}$ (the "meet-in-the-middle"\; algorithm).

The paper suggests a determined algorithm of solving the $rSUM$
problem for any $r\in\mathbb{N}$, which is of the order of $n\log^3
n$ in terms of the amount of memory used, with computational
complexity of the sub-quadratic order in some cases.

The idea of the obtained algorithm is based not considering integer
numbers, but rather $k\in\mathbb{N}$ successive bits of these
numbers in the binary numeration system. It is shown that if a sum
of integer numbers is equal to zero, then the sum of numbers
presented by any k successive bits of these numbers must be
sufficiently "close"\; (see Lemma~\ref{l2},~\ref{l3}) to zero. This
makes it possible to discard the numbers, which a fortiori, do not
establish the solution.

\section{Algorithm for solving the $rSUM$ problem}\label{s2}
Hereinafter, $|y|$ designates an absolute value of integer number
$y$, $\lceil y\rceil$  is the smallest integer greater than or equal
to $y$, $\lfloor y \rfloor$ is the smallest integer smaller than or
equal to $y$. A mapping $\text{sign}(y)$ returns the sign of integer
$y$ (it returns zero for zero).

Introduce mapping $P^k_j:\mathbb{Z}\mapsto\mathbb{F}^k_2$ for any
$k\in \mathbb{N}$ and $j\in \mathbb{N}\cup\{0\}$ as follows:
$$P^k_j(z) = \text{sign}(z)z_j,\;\forall z=\text{sign}(z)\sum\limits_{i=0}^{\infty}z_i2^{ik} \in
\mathbb{Z},$$ i.e. $j$ digit of integer $z$ in a numeral system with
base $2^k$.

Given:   set $\Omega$ of $n$ integer numbers, $m$ is the degree of a
polynomial, which bounds the maximum absolute value of input numbers
($n^m=2^{m\log_2 n})$.

\begin{alg}\label{a1}~ \par
1) From among the numbers in question, find $\zeta$, which is the
maximum in terms of its absolute value. Calculate $l =
\lceil\log_2(\zeta)\rceil.$

2) In a cycle on $j$ from $0$ to $\lfloor\frac{l+\lceil\log_2
r\rceil}{3\lceil\log_2 r\rceil}\rfloor$ perform the following:

~~~~2.1) Consider the numbers in $\Omega$ upon application of
$P^{3\lceil\log_2 r\rceil}_j$ and set them down in array $\Phi_j$ so
that the number of identical elements would not exceed $r$.

With each $\gamma \in \Phi_j$ group such ordinals of elements in
$\Omega$, where numbers with such ordinals in $\Omega$ and only
these numbers would be equal to $\gamma$ after using of
$P^{3\lceil\log_2 r\rceil}_j$. We associate it with table $\Pi_j$.

Brute force to find all $y_1\in \Phi_j$, where $\exists y_2, y_3,
\dots, y_{r}\in \Phi_j:$ $$|\sum\limits_{i=1}^{r}P^{3\lceil\log_2
r\rceil}_j (y_i)| < r \text{ mod } 2^{3\lceil\log_2 r\rceil},$$ for
$j = 0$, strict comparison to zero must be performed.

The gotten r-tuples, namely, their ordinals in $\Phi_j$, are to be
set down in $\Upsilon_j$.

3) Return $\Upsilon=\{$ $\Upsilon_j$ $\}$ and $\Pi=\{$ $\Pi_{j}$
$\}$.\end{alg}

\begin{alg}\label{a2} \verb"Algorithm for solving the rSUM problem"\par
1) Perform Algorithm~\ref{a1}: $\Upsilon^1,\;\Pi^1$.

2) Shift the elements of $\Omega$ cyclically by $\lceil\log_2
r\rceil$ bits to the right, that the sign bit is retained for all
numbers.

3) Perform Algorithm~\ref{a1} on conditions that for $j = 0$
inequality must be performed rather than comparison, and assume the
last $\lceil\log_2 r\rceil$ bits of numbers from $\Omega$ to be zero
bits: $\Upsilon^2,\;\Pi^2$.

4) Shift the elements of $\Omega$ cyclically by $\lceil\log_2
r\rceil$ bits to the right, that the sign bit is retained for all
numbers.

5) Perform Algorithm~\ref{a1} on conditions that for $j = 0$
inequality must be performed rather than comparison, and assume the
last $2\lceil\log_2 r\rceil$ bits of numbers from $\Omega$ to be
zero bits: $\Upsilon^3,\;\Pi^3$.

6) Shift the elements of $\Omega$ cyclically by $2\lceil\log_2
r\rceil$ bits to the left, that the sign bit is retained for all
numbers.

7) Return $\bigcap\limits_{i, j} \Upsilon^i_j$ relative to elements
of $\Omega$.
\end{alg}

We are now to prove that the presented algorithms are correct.

\begin{lemma}\label{l1}
For any $y_i \in \mathbb{Z}, i=1,\dots, r$, it is true that:

~~1) if $\sum\limits_{i=1}^{r}y_i=0,$ then
$\sum\limits_{i=1}^{r}y_i\equiv0 \text{ mod } 2^{k},$ where $k \in
\mathbb{N}.$

~~2) if $\sum\limits_{i=1}^{r}y_i\equiv0 \text{ mod } 2^{l}, l =
\max\limits_i(\lceil\log_2(|y_i|)\rceil+\lceil\log_2 r\rceil),$ then
$\sum\limits_{i=1}^{r}y_i=0.$\end{lemma}
\begin{proof}Obvious. This forms the basis of computer algebra.

The second statement is right because of $\sum\limits_{i=1}^{r} 2^t
= r2^t$.\end{proof}

\begin{lemma}\label{l2}
For any $y_i \in \mathbb{Z}, i=1,\dots, r$, it is true that:

if $\sum\limits_{i=1}^{r}y_i=0,$ then
$|\sum\limits_{i=1}^{r}P^k_j(y_i)|< r\text{ mod }
2^k,\\j=0,\dots,\lfloor\frac{l}{k}\rfloor,$ $l =
\max\limits_i(\lceil\log_2(y_i)\rceil+\lceil\log_2 r\rceil),$
$k>\lceil\log_2 r\rceil\in \mathbb{N}$.\end{lemma}
\begin{proof}For $j=0$ the condition of Lemma~\ref{l2} is met by virtue of Lemma~\ref{l1}.

Assume the opposite meaning that for a value $j = s$, for some $r$
numbers meeting the condition of Lemma~\ref{l2}, the required
inequality is wrong. At the same time, by virtue of Lemma~\ref{l1}:
$$\sum\limits_{i=1}^{r}y_i \equiv0 \text{ mod } 2^{sk}.$$

Present each $y_i\text{ mod } 2^{sk}$ as a sum of the value $P^k_s$
(the last $k$ bits of numbers $\text{sign}(y_i)(|y_i| \text{ mod }
2^{sk}$)) and the residue by module $2^{(s-1)k}$, then
$$2^{(s-1)k}\sum\limits_{i=1}^{r}P^k_s(y_i)\equiv
-(\sum\limits_{i=1}^{r}\text{ sign}(y_i)(|y_i| \text{ mod }
2^{(s-1)k})) \equiv \delta2^{(s-1)k}\text{ mod } 2^{sk},$$ where
$|\delta|< r$, as the sum of $r$ numbers, the absolute value of
which is smaller than $2^j$ for a natural $j$, cannot exceed
$r2^j-r$. Besides, we know from Lemma~\ref{l1} that
$\sum\limits_{i=1}^{r}y_i \equiv0 \text{ mod } 2^{(s-1)k}$. From
here, we obtain the required.\end{proof}

\begin{lemma}\label{l3}
For any $y_i \in \mathbb{Z}, i=1,\dots, r$, it is true that:

if $\sum\limits_{i=1}^{r}y_i=0,$ then for $\tilde{y}_i$ the
inequality $|\sum\limits_{i=1}^{r}P^k_j(\tilde{y}_i)|< r\text{ mod }
2^k$ is true, where $\tilde{y}_i$ is obtained from $y_i$ by
arithmetic shift to the right by $t$ bits.

$t, k>\lceil\log_2 r\rceil$ are any natural numbers, and $j$ is any
non-negative integer.\end{lemma}
\begin{proof}$$2^{t+k(j-1)}\sum\limits_{i=1}^{r}P^k_j(\tilde{y}_i) \equiv
-(\sum\limits_{i=1}^{r} \text{ sign}(y_i) (|y_i| \text{ mod }
2^{t+k(j-1)})) \text{ mod } 2^{t+kj}.$$

Further on, the proof totally replicates the proof of
Lemma~\ref{l2}.\end{proof}

\begin{theorem}\label{t1} Algorithm~\ref{a2} will issue the solution of the $rSUM$ problem.\end{theorem}
\begin{proof}As follows from Lemmas~\ref{l1},~\ref{l2},~\ref{l3}, if there exists a solution of the $rSUM$ problem then, after execution of Algorithm~\ref{a2}, and even more so after execution of Algorithm~\ref{a1}, these numbers will stay within $\Omega$.

The cycle on $j$ in Algorithm~\ref{a1} finishes at iteration
$\lfloor\frac{l+\lceil\log_2 r\rceil}{3\lceil\log_2 r\rceil}\rfloor$
by virtue of the second if-clause in Lemma~\ref{l1}.

After step 1), for each $y_1, y_2, \dots, y_{r} \in \Omega$ takes
place $|\sum\limits_{i=1}^{r}P^{3\lceil\log_2 r\rceil}_j(y_i)|< r
\text{ mod } 2^{3\lceil\log_2 r\rceil}$ for any $j$ under
consideration, for $j = 0$ comparison to zero is performed.

It is about the numbers as such, not some values of
$P^{3\lceil\log_2 r\rceil}_j$ of various numbers at each step on
$j$; this is why we remembered ordinals in r-tuples for --- to
coincide at each step of cycle $j$.

Hence $$\sum\limits_{i=1}^{r}y_i =
\sum\limits_{i=1}^{\lfloor\frac{l+\lceil\log_2
r\rceil}{3\lceil\log_2 r\rceil}\rfloor}z_i2^{3i\lceil\log_2
r\rceil}, \text{ where } |z_i|< 2r-1,$$ as, considering $y_i$ after
using of $P^{3\lceil\log_2 r\rceil}_j$, we may lose in
$\sum\limits_{i=1}^{r}P^{3\lceil\log_2 r\rceil}_j(y_i)$ $r-1$ carry
bits by absolute value relative to the sum $P^{3\lceil\log_2
r\rceil}_j(\sum\limits_{i=1}^{r}y_i)$ (see the proof in
Lemma~\ref{l2}); besides, the very inequality from Lemma~\ref{l2}
makes it possible to differentiate from zero by absolute value to $r
- 1$.

Yet, at step 3), the sum  $P^{3\lceil\log_2 r\rceil}_j$ of
$\tilde{y}_1, \tilde{y}_2, \dots, \tilde{y}_{r}$, where
$\tilde{y}_i$ is $y_i$ at step 2) cyclically shifted to the right by
$\lceil\log_2 r\rceil$, will not meet the necessary inequality for
module $2^{3\lceil\log_2 r\rceil}$ (see Lemma~\ref{l3}) for the
first $j: z_j \neq 0$, if $z_j<r$, as in the latter case, this $z_j$
will not be constituted by the least significant $\lceil\log_2
r\rceil$ bits of a $3\lceil\log_2 r\rceil$-bit number in the binary
numeral system, but by more significant bits, which is determined by
the fact that
$$\sum\limits_{i=1}^{r}\tilde{y}_i = t +
\sum\limits_{i=1}^{\lfloor\frac{l+\lceil\log_2
r\rceil}{3\lceil\log_2 r\rceil}\rfloor}z_i2^{3i\lceil\log_2
r\rceil-\lceil\log_2 r\rceil}, \text{ where } |t|< r.$$

The correctness of this presentation of the sum $\tilde{y}_i$
follows from ideas presented in Lemmas~\ref{l2},~\ref{l3}, as, with
a cyclic shift of numbers $y_i$, we may lose $r-1$ carry bits by
absolute value.

At step 5) we will exclude these $y_1, \dots, y_{r}$, if the first
$z_j \neq 0$ is larger than $r-1$, for the same
considerations.\end{proof}

\section{Computational complexity of suggested algorithm}\label{s3}

\begin{lemma}\label{l4} Algorithm's~\ref{a1} order of complexity is
$n\log n.$\end{lemma}
\begin{proof} Calculating the maximum element by absolute value is $n$ operations.

Applying $P^{3\lceil\log_2 r\rceil}_j$ to elements of $\Omega$ is no
more than $2n$ operations (taking in modulus and cyclic shift).
Adding the obtained values to $\Phi_j$ after applying of
$P^{3\lceil\log_2 r\rceil}_j$, containing no more $r$ identical
elements, using insertion sort with binary search, is not more than
$n(r2^{3\lceil\log_2 r\rceil}+4\lceil\log_2 r\rceil)$ operations,
where we use $4\lceil\log_2 r\rceil$ to assess the complexity of
binary search, $r2^{3\lceil\log_2 r\rceil}$ is the number of shifts
of elements in an array for insertion to a proper place.

At step 2.1) we solve the $rSUM$ problem by modulus
$2^{3\lceil\log_2 r\rceil}$ for a quantity of different numbers not
exceeding $r2^{3\lceil\log_2 r\rceil}$, though there may be more
than one solution. The exhaustive enumeration of all the variants
requires $r^r2^{3r\lceil\log_2 r\rceil}$ operations.

All the above-calculated was a single iteration on cycle of $j$.

As $l = m\lceil\log_2 n\rceil + \lceil\log_2 r\rceil$ and $r,\;m$
are fixed numbers, we obtain the required assessment.\end{proof}

\begin{remark}\label{r3}It is convenient to assume that each element in the r-tuple
 from $\Upsilon_{j}$ (where elements of the r-tuple are ordinals of elements in
$\Phi_{j}$, as determined by us) is a column of such ordinals of
elements in $\Omega$, that the numbers corresponding to these
ordinals in $\Omega$ upon application of  $P^{3\lceil\log_2
r\rceil}_{j}$ will be equal to an element with this ordinal. We may
assume so, because we have a table of association of the elements in
$\Phi_{j}$ with elements in $\Omega$.\end{remark}

\begin{theorem}\label{t2} Algorithm's~\ref{a2} order of complexity is
sub-quadratic for some cases.
\end{theorem}
\begin{proof}All steps of the Algorithm~\ref{a2} except step 7) do not exceed the $n\log n$ order (see Lemma~\ref{l4}).

How to compute $\bigcap\limits_{i, j} \Upsilon^i_j$ relative to
elements of $\Omega$?

All r-tuples from $\Upsilon^i_j$ are tables, see Remark~\ref{r3}.

$\Upsilon^i_j$ contains no more $2r!r2^{3\lceil\log_2 r\rceil(r-1)}$
items. Comparing a r-tuple with another according to ordinals in
$\Omega$ will not make more than $rn\log_2 n$ operations. Consider
$\log_2 n$ as elements in $\Omega$ are read successively, and hence,
ordinals of elements of $\Omega$, related to an element of $\Phi_j$,
are set down in an orderly way, which means that we may use binary
search. Every time we create new r-tuple with common ordinals of
$\Omega$ in columns in one r-tuple and the other, if there is at
least one common element in each column.

As cycle $j$ ends $\lceil\frac{m\lceil\log_2 n\rceil}{3\lceil\log_2
r\rceil}\rceil$ in Algorithm~\ref{a1} and there are 3 execution of
Algorithm~\ref{a1} in Algorithm~\ref{a2}, we get upper bound of
vertices of such comparing r-tuples tree:
$$(2r!r2^{3(r-1)\lceil\log_2 r\rceil})^{\lceil\frac{m\lceil\log_2
n\rceil}{\lceil\log_2 r\rceil}\rceil}.$$ It's a lot,
 that's why we compute $$\Gamma_s = \bigcap\limits_{i,\; j=sh, \dots, (s+1)h-1}
 \Upsilon^i_j,\; where\; i=1, 2, 3,\; h = \lceil\frac{\lceil\log_2\log_2 n\rceil}{9r\lceil\log_2 r\rceil}\rceil,\; s=0, \dots, \lceil\frac{m\lceil\log_2 n\rceil}{3h\lceil\log_2
r\rceil}\rceil.$$

Cardinality of $\Gamma_s$ is less than
$$(2r!r2^{3\lceil\log_2
r\rceil(r-1)})^{\lceil\frac{\lceil\log_2\log_2
n\rceil}{3r\lceil\log_2 r\rceil}\rceil} \leq \log_2^2 n.$$ So, the
order of complexity of the computation of all $\Gamma_s$ is less
than $n\log_2^{3} n$.

Find $\lceil\frac{\log_{\lceil\log_2 n\rceil} n}{3}\rceil$ sets
$\Gamma_s$ with the smallest number of elements (it is of the order
of $n\log n$ operation) and compute confluence of them $\Theta$ (it
is of the order of $n^{\frac{5}{3}}\log_2 n$ operations).

To count the quantity of all variants produced by each r-tuple from
$\Theta$, relative to elements of $\Omega$, takes no more than
$2rn^{\frac{5}{3}}$ operations (amount of options generated by fixed
r-tuple is the product of the number of items in a columns of this
r-tuple).

\textbf{If the total number of r-tuples} from $\Theta$, relative to
elements of $\Omega$, is less than $n^{\frac{3}{2r}}$, we get
sub-quadratic time for our algorithm (brute force all of variants).

\textbf{If the total number of r-tuples} from $\Theta$, relative to
elements of $\Omega$, is less than
$\frac{n^{\frac{1}{2}}}{\log^{\frac{1}{r}} n}$, brute force still
would be faster than using known algorithms.
\end{proof}

\begin{theorem}\label{t3} Algorithm~\ref{a2} requires an amount of memory of an order $n\log^3 n$ relative to storage of integers.\end{theorem}
\begin{proof}As will readily be observed, the most memory-consuming step is
7).

Step 7) of Algorithm~\ref{a2} requires some memory for
$\Upsilon^i_{j}$ (constant quantity) and $\Pi^i_j$ associating
elements in $\Upsilon^i_{j}$ with elements in $\Omega$ (not more
than the order of $n$), $i=1, 2, 3$, $j=0, \dots, m\log n +\log
r$.

All together $\Gamma_s$ require the order of $n\log^3 n$ memory, see
Theorem~\ref{t2}.\end{proof}

\begin{remark}\label{r5}What is it about the constant in asymptotic complexity?

As follows from Theorem~\ref{t2} and Lemma~\ref{l4} the constant
would not exceed $3mr^{4r}.$\end{remark}

\begin{remark}\label{r6} As time and memory complexity of suggested algorithm is of the sub-quadratic order, it seems to be useful to perform it at the beginning of any other known algorithm.\end{remark}

\end{fulltext}

\end{document}